\documentclass[1p,final]{elsarticle}
\usepackage{lipsum}

\makeatletter
\def\ps@pprintTitle{%
   \let\@oddhead\@empty
   \let\@evenhead\@empty
   \def\@oddfoot{\reset@font\hfil\thepage\hfil}
   \let\@evenfoot\@oddfoot
}
\makeatother

\usepackage{amsfonts,color,morefloats,pslatex}
\usepackage{amssymb,amsthm, amsmath,latexsym}

\newtheorem{theorem}{Theorem}
\newtheorem{lemma}[theorem]{Lemma}

\newtheorem{corollary}[theorem]{Corollary}

\newtheorem{open}[theorem]{Open Problem}

\newtheorem{example}[theorem]{Example}

\newtheorem{conj}[theorem]{Conjecture}

\newcommand{\rank}{{\mathrm{rank}}}

\newcommand{\tr}{{\mathrm{Tr}}}

\newcommand{\gf}{{\mathrm{GF}}}
\newcommand{\PG}{{\mathrm{PG}}}

\newcommand{\support}{{\mathrm{suppt}}}
\newcommand{\Aut}{{\mathrm{Aut}}}
\newcommand{\PAut}{{\mathrm{PAut}}}
\newcommand{\MAut}{{\mathrm{MAut}}}
\newcommand{\GAut}{{\mathrm{Aut}}}

\newcommand{\Sym}{{\mathrm{Sym}}}

\newcommand{\GL}{{\mathrm{GL}}}

\newcommand{\wt}{{\mathtt{wt}}}

\newcommand{\RM}{{\mathrm{RM}}}
\newcommand{\PRM}{{\mathrm{PRM}}}

\newcommand{\cP}{{\mathcal{P}}}
\newcommand{\cS}{{\mathcal{S}}}
\newcommand{\cB}{{\mathcal{B}}}

\newcommand{\C}{{\mathcal{C}}}

\newcommand{\cH}{{\mathcal{H}}}

\newcommand{\ba}{{\mathbf{a}}}
\newcommand{\bb}{{\mathbf{b}}}
\newcommand{\bc}{{\mathbf{c}}}
\newcommand{\bg}{{\mathbf{g}}}

\newcommand{\bx}{{\mathbf{x}}}

\newcommand{\bzero}{{\mathbf{0}}}

\newcommand{\bD}{{\mathbb{D}}}

\newcommand{\GA}{{\mathrm{GA}}}

\newcommand{\PGL}{{\mathrm{PGL}}}

\usepackage{blindtext}

\begin{document}

\begin{frontmatter}




\title{The Support Designs of Several Families of Lifted Linear Codes}


\author[cding]{Cunsheng Ding}
\ead{cding@ust.hk}

\author[zhs]{Zhonghua Sun}
\ead{sunzhonghuas@163.com}

\author[cding]{Qianqian Yan}
\ead{19118010@bjtu.edu.cn, yanqianqian@ust.hk}

\address[cding]{Department of Computer Science and Engineering, The Hong Kong University of Science and Technology, Clear Water Bay, Kowloon, Hong Kong, China}

\address[zhs]{School of Mathematics, Hefei University of Technology, Hefei, 230601, Anhui, China}


\begin{abstract} 

A generator matrix of a linear code $\C$ over $\gf(q)$ is also a matrix of the same rank $k$ over any 
extension field $\gf(q^\ell)$ and generates a linear code of the same length, same dimension and 
same minimum distance over $\gf(q^\ell)$, denoted by $\C(q|q^\ell)$ and called a lifted code of $\C$.  Although $\C$ and 
their lifted codes $\C(q|q^\ell)$ have the same parameters, they have different weight distributions 
and different applications. Few results about lifted linear codes are known in the literature. 
This paper proves some fundamental theory for lifted linear codes, 
and studies the support $2$-designs of the lifted projective Reed-Muller codes, lifted Hamming codes and lifted Simplex codes. In addition, this paper 
settles the weight distributions of  the lifted Reed-Muller codes of certain orders, 
and investigates the support $3$-designs of these  lifted codes. 
As a by-product,  an infinite family of three-weight projective codes over $\gf(4)$ is obtained.  
\end{abstract}


\end{frontmatter}


\section{Introduction}

\subsection{The support designs of linear codes}

Let $v$ and $k$ be positive integers with $1 \leq k \leq v$.  Let $t$ be a positive integer with $t \leq k$.  
Let $\cP$ be a set of $v$ elements and let $\cB$ be a set of $k$-subsets of $\cP$. The pair $\bD := (\cP, \cB)$ is an 
incidence structure, where the incidence relation is the set membership. The pair  
$\bD = (\cP, \cB)$ is called a $t$-$(v, k, \lambda)$ {\em design}, or simply {\em $t$-design}, if every $t$-subset of $\cP$ is contained in exactly $\lambda$ elements of
$\cB$. The elements of the set $\cP$ are referred to as {\em points}, and those of the set $\cB$ are called {\em blocks}. 
The set $\cB$ is called the block set. 
The number of blocks in $\cB$ is usually denoted by $b$. 
Let $\binom{\cP}{k}$ denote the set of all $k$-subsets of $\cP$. Then the incidence structure $\left (\cP, \binom{\cP}{k} \right )$ is a $k$-$(v, k, 1)$ design, 
which is called a \emph{complete design}. 
  A $t$-design is said to be {\em simple} if $\cB$ does not contain
any repeated blocks.  This paper 
considers only simple $t$-designs.
A $t$-$(v,k,\lambda)$ design is called a
{\em Steiner system} if $t \geq 2$ and $\lambda=1$,
and is denoted by $S(t,k, v)$. 

There are several approaches to constructing $t$-designs with linear codes.  
One of the coding-theoretic constructions of $t$-designs is the following.   
Let $\C$ be a $[v, \kappa, d]$ linear code over $\gf(q)$, which 
is a $\kappa$-dimensional subspace of the vector space $\gf(q)^v$ with minimum Hamming distance $d$ \cite{HP03,MS77} 
and is also called a $[v, \kappa, d]_q$ code.  We sometimes use $d(\C)$ to denote the minimum Hamming 
distance of a linear code $\C$. 
Let $A_i$ or $A_i(\C)$ denote the number of codewords with Hamming weight $i$ in $\C$ 
for $0 \leq i \leq v$. The sequence $(A_0, A_1, \ldots, A_v)$ is called the \emph{weight distribution} of $\C$.  
The {\it support} of a codeword  $\bc=(c_1,  c_2, \ldots, c_v)$ of a linear code $\C$ is defined by 
$$
\support(\bc)=\{1 \leq i \leq v: c_i \neq 0\}. 
$$ 
For each $k$ with $A_k(\C) \neq 0$,  let $\cB_k(\C)$ denote
the set of the supports of all codewords with Hamming weight $k$ in $\C$, 
where $\cB_k(\C)$ is not allowed to have repeated blocks. 
Suppose that the coordinates of the codewords in $\C$  
are indexed by $(p_1, \ldots, p_v)$. Let $\cP(\C)=\{p_1, \ldots, p_v\}$.  The incidence structure $(\cP(\C), \cB_k(\C))$
may be a $t$-$(v, k, \lambda)$ simple design for some positive integer $\lambda$, which is called a
\emph{support design} of the code $\C$, and is denoted by $\bD_k(\C)$. In this case, we say that the codewords of weight $k$ 
in $\C$ support or hold a $t$-$(v, k, \lambda)$ design, and for simplicity, we say that $\C$ supports or holds a $t$-$(v, k, \lambda)$ design. 

The following theorem, which was established by Assmus and Mattson and called the Assmus-Mattson Theorem, 
says that the incidence structure $\bD_k(\C):=(\cP(\C), \cB_k(\C))$ defined by 
a linear code $\C$ is a simple $t$-design under certain conditions \cite{AM69}.

\begin{theorem}\label{thm-designAMtheorem}
Let $\C$ be a $[v,k,d]_q$ code.  Let $d^\perp$ denote the minimum distance of the dual code 
$\C^\perp$ of $\C$. 
Let $w$ be the largest integer satisfying $w \leq v$ and 
$$ 
w-\left\lfloor  \frac{w+q-2}{q-1} \right\rfloor <d. 
$$ 
Define $w^\perp$ analogously using $d^\perp$. Let $(A_i)_{i=0}^v$ and $(A_i^\perp)_{i=0}^v$ denote 
the weight distribution of $\C$ and $\C^\perp$, respectively. Fix a positive integer $t$ with $t<d$, and 
let $s$ be the number of $i$ with $A_i^\perp \neq 0$ for $1 \leq i \leq v-t$. Suppose $s \leq d-t$. Then 
\begin{itemize}
\item $\bD_i(\C)$ is a $t$-design provided $A_i \neq 0$ and $d \leq i \leq w$, and 
\item $\bD_i(\C^\perp)$ is a $t$-design provided $A_i^\perp \neq 0$ and 
         $d^\perp \leq i \leq \min\{v-t, w^\perp\}$. 
\end{itemize}
\end{theorem}

The preceding Assmus-Mattson Theorem is a very useful tool for constructing $t$-designs from linear codes 
(see, e.g., \cite{DTbook22}, \cite{DWF}, \cite{HW23}, \cite{Tonchev1}, \cite{Tonchev2}, \cite{XCQ22}), but does not 
characterize all linear codes supporting $t$-designs. The reader is referred to \cite{TDX19} for a generalized 
Assmus-Mattson theorem.

The second sufficient condition for the incidence structure $(\cP(\C), \cB_k(\C))$ to be a $t$-design is via the 
automorphism group of the linear code $\C$. Before introducing this  sufficient condition, we have to recall 
several different automorphism groups of a linear code.

The set of coordinate permutations that map a code $\C$ to itself forms a group,  where the binary operation 
for this group is the function composition. This group 
is referred to as the \emph{permutation automorphism group} of $\C$
and denoted by $\PAut(\C)$. If the length of $\C$ is $n$ and the coordinates of the codewords in $\C$ are indexed 
with the elements in the set $\{1,2, \ldots,n\}$, then $\PAut(\C)$ is a subgroup of the
\emph{symmetric group} $\Sym_n$.

A \emph{monomial matrix} over $\gf(q)$ is a square matrix that has exactly one
nonzero element of $\gf(q)$  in each row and column. It is easily seen that a monomial matrix $M$ can be written either in
the form $DP$ or the form $PD_1$, where $D$ and $D_1$ are both diagonal matrices and $P$ is a permutation
matrix. Clearly, the set of monomial matrices that map $\C$ to itself forms a group denoted by $\MAut(\C)$,  which is called the
\emph{monomial automorphism group\index{monomial automorphism group}} of $\C$. By definition, 
we have 
$$
\PAut(\C) \subseteq \MAut(\C).
$$
By definition, every element in $\MAut(\C)$ is of the form $DP$, where $D$ is a diagonal matrix and 
$P$ is a permutation matrix. 
An element of this group acts on a codeword $\bc$ as 
$$
(DP)(\bc)=\bc DP^{-1}. 
$$ 
For two elements $(D_i, P_i):=D_i P_i$ in  $\MAut(\C)$ for $i \in \{1,2\}$, 
the corresponding binary operation of the group $\MAut(\C)$ is defined by  
$$
(D_2, P_2) \circ (D_2, P_1)=
(D_1 P_1^{-1} D_2P_1, P_2P_1).
$$

The \textit{automorphism group}\index{automorphism group} of $\C$, denoted by $\GAut(\C)$, is the set
of maps of the form $M\gamma$,
where $M$ is a monomial matrix and $\gamma$ is a field automorphism, that map $\C$ to itself. 
For binary codes $\C$,  $\PAut(\C)$,  $\MAut(\C)$ and $\GAut(\C)$ are the same. If $q$ is a prime, $\MAut(\C)$ and
$\GAut(\C)$ are identical. According to their definitions, we have the following relations: 
$$
\PAut(\C) \subseteq \MAut(\C) \subseteq \GAut(\C).
$$ 
By definition, every element in $\GAut(\C)$ is of the form $DP\gamma$, where $D$ is a diagonal matrix,
$P$ is a permutation matrix, and $\gamma$ is an automorphism of $\gf(q)$. 
An element of this group acts on a codeword $\bc$ as 
$$
(DP\gamma)(\bc)=\gamma((\bc D)P^{-1}). 
$$ 
For any two elements $(D_i, P_i, \gamma_i):=D_i P_i \gamma_i$ in  $\GAut(\C)$ for $i \in \{1,2\}$, 
the corresponding binary operation of the group $\GAut(\C)$ is defined by  
$$
(D_2, P_2, \gamma_2) \circ (D_1, P_1, \gamma_1)=
(D_1P_1^{-1} \gamma_1^{-1}(D_2)P_1, P_2P_1, \gamma_2 \circ \gamma_1).
$$

The automorphism group $\GAut(\C)$ of a linear code $\C$ is said to be {\em $t$-transitive} if for every pair of $t$-element ordered
sets of coordinates, there is an element $DP\gamma$ of the automorphism group $\GAut(\C)$ such that its
permutation part $P$ sends the first set to the second set. The automorphism group $\GAut(\C)$ is said to be {\em $t$-homogeneous} if for every pair of $t$-element 
sets of coordinates, there is an element $DP\gamma$ of the automorphism group $\GAut(\C)$ such that its
permutation part $P$ sends the first set to the second set.

With the help of the automorphism group of a linear code $\C$, the following theorem gives another sufficient condition 
for the code $\C$ to hold $t$-designs.

\begin{theorem}~\cite[p. 308]{HP03}\label{thm-designCodeAutm}
Let $\C$ be a linear code of length $n$ over $\gf(q)$ such that $\GAut(\C)$ is $t$-transitive 
or $t$-homogeneous. Then the codewords of any weight $i \geq t$ of $\C$ hold a $t$-design, 
i.e., $\bD_i(\C)$ is a $t$-design for each $i \geq t$ with $A_i(\C)>0$.
\end{theorem} 

It is in general very hard to determine the full automorphism group of a linear code.  
As long as a $t$-transitive or $t$-homogeneous subgroup of $\GAut(\C)$ is found, 
one could make use of Theorem \ref{thm-designCodeAutm} to prove the $t$-design 
property of $\bD_i(\C)$. 

The third way for proving the $t$-design property of an incidence structure $\bD_i(\C)$ is the direct approach \cite{DingTang19,TD20,YZ}, 
where the $t$-design property of an incidence structure $(\cP(\C), \cB_i(\C))$ is proved by verifying the conditions in the definition of $t$-designs directly. This direct approach may work only when the block size is very small.

\subsection{Lifted linear codes}

Let $q$ be a power of a prime and let $\C$ be an $[n, k, d]_q$ linear code with generator matrix $G$. 
For a positive integer $\ell$, $G$ is also a matrix of rank $k$ over $\gf(q^\ell)$. Let $\C(q|q^\ell)$ denote the linear code over $\gf(q^\ell)$ 
generated by $G$, which is called the \emph{lifted code over $\gf(q^\ell)$} of $\C$. It is known that $\C(q|q^\ell)$ and 
$\C$ have the same length, dimension and minimum distance \cite[Theorem 7]{DT21},  but they have different weight distributions for $\ell>1$ (see the examples and some general results in later sections). 
Even if the weight 
distribution of $\C(q)$ is known, it may be very hard to determine the weight distribution of a lifted code 
$\C(q|q^\ell)$ for $\ell>1$.  By definition, $\C$ is the subfield subcode over $\gf(q)$ of the lifted code 
$\C(q|q^\ell)$ \cite{HP03,MS77}.

\subsection{The idea of constructing new designs with lifted codes of a linear code} 

Suppose that $\C$ supports some nontrivial $t$-designs.  
One or both of the following cases may happen: 
\begin{itemize}
\item  $\bD_i(\C(q|q^\ell))$ is a $t'$-$(n, i, \lambda'_i)$ design with $(t', \lambda'_i)\neq (t, \lambda_i)$,
            while $\bD_i(\C)$ is a $t$-$(n, i, \lambda_i)$ design for some $i$ with $A_i(\C) \neq 0$.   
 \item  $\bD_i(\C(q|q^\ell))$ is a $t'$-$(n, i, \lambda'_i)$ design with $A_i(\C(q|q^\ell)) \neq 0$,
            while $A_i(\C) = 0$.   
\end{itemize}   
If any of the two cases above happens, a new design could be obtained. This idea of obtaining new support 
designs from lifted linear codes was considered in \cite{Ding241}, where some $5$-designs were found.   

The following theoretical result will be needed in subsequent sections and is a theoretical foundation of this paper. 

\begin{theorem}\label{thm-main} \cite{Ding241}
Let $\C$ be a linear code of length $n$ over $\gf(q)$. 
Let $E$ be a subgroup of the monomial automorphism group $\MAut(\C)$. Assume that $E$ is $t$-transitive 
or $t$-homogeneous. Then the following hold: 
\begin{enumerate}
\item For each $e \in E$, $e(\C(q|q^\ell))=\C(q|q^\ell)$. 
\item $\bD_i(\C(q|q^\ell))$ is a $t$-$(n, i, \lambda_i)$ design with $A_i(\C(q|q^\ell)) \neq 0$ for any $i \geq t$, where $\lambda_i$ 
is an integer. 
\end{enumerate} 
\end{theorem}

\subsection{Motivations and objectives of this paper} 

Many infinite families of $2$-designs and $3$-designs supported by linear 
codes are available in the literature \cite{DTbook22,Tonchev1,Tonchev2}. 
Recently, several infinite 
families of $4$-designs from linear codes were reported in \cite{TD20,YZ}.   
No infinite family of linear codes supporting an infinite family of nontrivial simple $5$-designs 
is known.  
It is not easy to construct $5$-designs from linear codes. 
A very small number of $5$-designs from some linear codes were reported in \cite{AM69}, 
\cite{Pless72},   \cite[Appendix A]{DTbook22}, and \cite{Ding241}. 
Until now no linear code supporting a nontrivial $6$-design is known in the literature. 
Of course, there are algebraic, geometric and combinatorial approaches to constructing $t$-designs (\cite{AK92}, \cite{BJL}).   

Combinatorists are in general interested only in the following types of $t$-designs: 
\begin{itemize}
\item $t$-$(v, k, \lambda)$ designs with large strength $t$. 
\item $t$-$(v, k, \lambda)$ designs with small value $\lambda$,  in particular, $\lambda=1$. 
\item Symmetric and quasi-symmetric $t$-designs.    
\end{itemize}
However,  coding theorists are also interested in $t$-$(v, k, \lambda)$ designs with small strength $t$ and large $\lambda$, as some $2$-designs could be used to construct linear codes with very good parameters \cite{DT21}. This is the main motivation of obtaining new $2$-designs by studying the lifted Hamming and Simplex codes in this paper. 
The second motivation of this paper is that two-weight and three-weight linear codes have interesting 
applications in association schemes, cryptography and graph theory.  
  
As a follow-up of \cite{Ding241}, we will do the following in this paper:
\begin{itemize}
\item Add some new fundamental results for lifted linear codes. 
\item Study the $2$-designs supported by lifted projective Reed-Muller codes.  
\item Study the $2$-designs supported by lifted Hamming codes and lifted Simplex codes.  
\item Settle the weight distributions of the lifted codes of the Reed-Muller codes of some orders. 
\item Present many infinite families of $3$-designs supported by lifted Reed-Muller codes. 
\end{itemize} 
The contributions of this paper are summarised in Section \ref{sec-last}.  The results of this paper complement 
the literatures of coding theory and combinatorics.

\section{Fundamental results of lifted linear codes} 

In this section, we prove some fundamental results for lifted linear codes. The following theorem extends earlier results about lifted linear codes \cite{DT21}. 

\begin{theorem}\label{thm-fundam2}
Let $n $ and $\ell$ be positive integers and let $\C$ be an $[n, k, d]_q$ linear code with $k \geq 1$. Then the following hold.
\begin{enumerate}
\item If $\{\alpha_1, \alpha_2,\ldots, \alpha_\ell \}$ is a basis of $\gf(q^\ell)$ over $\gf(q)$, then the lifted code
$$\C(q | q^\ell)=\{\alpha_1 \bc_1+ \alpha_2 \bc_2+\cdots+  \alpha_\ell \bc_\ell:~\bc_i\in \C,~1\leq i\leq \ell \}.$$ 
\item The lifted code $\C(q|q^\ell)$ has parameters $[n, k, d]_{q^\ell}$. Furthermore, $$
A_d(\C(q|q^\ell))=\frac{q^\ell-1}{q-1} A_d(\C),
$$ 
and every minimum weight codeword $\bc(q|q^\ell)$ in $\C(q|q^\ell)$ is of the form $\bc(q|q^\ell)=u \bc$, where $\bc$ is a minimum weight codeword in $\C$ and $u \in \gf(q^\ell)^*$.
\item If $G=[\bg_1^T \ldots \bg_n^T]$ is a generator matrix of $\C$ and $\C$ is projective, then
$$A_w(\C(q|q^\ell) )=\left|\left \{ B \in \gf(q)^{\ell \times k}:\left|V_B \cap \{\bg_1, \ldots ,\bg_n\}\right|=n-w   \right \}\right|,$$
where $V_B=\{ \bg \in \gf(q)^k: B \bg^T =\bzero^T  \}$.
\end{enumerate}
\end{theorem}

\begin{proof}
Let $G$ be a generator matrix of $\C$. Then $\bc \in \C(q| q^\ell)$ if and only if $\bc=\bb G$, where $\bb \in \gf(q^\ell)^k$. For any $\bb \in \gf(q^\ell)^k$, there are $\bb_1,\bb_2, \ldots, \bb_\ell \in \gf(q)^k$ such that
$$\bb=\alpha_1 \bb_1+\alpha_2\bb_2+\cdots + \alpha_\ell \bb_\ell.
$$
It follows that $\bc=\alpha_1\bc_1+\alpha_2\bc_2+\cdots+\alpha_\ell \bc_\ell$, where $\bc_i=\bb_i G \in \C$. It is straightforward to see that
 $$\C(q | q^\ell)=\{\alpha_1 \bc_1+ \alpha_2 \bc_2+\cdots+  \alpha_\ell \bc_\ell:~\bc_i\in \C,~1\leq i\leq \ell \}.$$ The desired Result 1 follows.
 
For any $\bc=\alpha_1 \bc_1+ \alpha_2 \bc_2+\cdots+  \alpha_\ell \bc_\ell$, where $\bc_i=(c_{i1}, c_{i2},\ldots,c_{in})\in \C$, it is easily verified that 
 \begin{align}\label{eq:1}
 \wt( \bc)&=|{\rm Supp}(\bc_1)\cup {\rm Supp}(\bc_2) \cup \cdots \cup {\rm Supp}(\bc_\ell)|\\
& \geq \max_{1\leq i\leq \ell} \{\wt(\bc_i)\}\geq d, \notag
 \end{align}
where ${\rm Supp}(\bc_i):=\{j: ~c_{ij}\neq 0,~1\leq j\leq n \}$. Below we prove that $\wt(\bc)=d$ if and only if there exist $\bb \in \C$ with $\wt(\bb)=d$ and $u \in \gf(q^\ell)^*$ such that $\bc=u \bb$.  

Notice that $|{\rm Supp}(\bc_i)|\geq d$ for any $\bzero \neq \bc_i \in \C$. It follows from (\ref{eq:1}) that $\wt( \bc)=d$ if and only if there are $i_1, i_2, \ldots, i_t\in \{1, 2, \ldots, \ell\}$ such that 
 \begin{equation}\label{eq:2}
 {\rm Supp}(\bc_{i_1})={\rm Supp}(\bc_{i_2})=\cdots = {\rm Supp}(\bc_{i_t})	
 \end{equation} 
 with $|{\rm Supp}(\bc_{i_1})|=d$ and ${\rm Supp}(\bc_j)=\emptyset$ for $j \notin \{i_1, i_2, \cdots, i_t\} $, where $1\leq t\leq \ell$. 
\begin{itemize}
\item If $t=1$,	$\bc=\alpha_{i_1} \bc_{i_1}$. The desired result follows.
\item If $t\geq 2$, it follows from (\ref{eq:2}) and $|{\rm Supp}(\bc_{i_1})|=|{\rm Supp}(\bc_{i_2})|=\cdots=|{\rm Supp}(\bc_{i_t})|=d$ that there are $\lambda_2,\ldots,\lambda_t \in \gf(q)^*$ such that $\bc_{i_j}=\lambda_j \bc_{i_1}$ for $2\leq j\leq t$. Consequently, 
$$\bc=(\alpha_{i_1}+\lambda_2 \alpha_{i_2} +\cdots+\lambda_t \alpha_{i_t}) \bc_{i_1}.$$
\end{itemize}
In summary, $\bc \in \C(q|q^\ell)$ with $\wt(\bc)=d$ if and only if there exist $\bb \in \C$ with $\wt(\bb)=d$ and $u \in \gf(q^\ell)^*$ such that $\bc=u \bb$. The desired Result 2 follows. 

For any $\bc=\alpha_1 \bc_1+ \alpha_2 \bc_2+\cdots+  \alpha_\ell \bc_\ell$, where $\bc_i=(c_{i1}, c_{i2},\ldots,c_{in})\in \C$, it is also easily verified that
\begin{align}\label{eq:3}
\wt(\bc)&=n-|{\rm Supp}(\bc_1) \cap {\rm Supp}(\bc_2) \cap \cdots \cap {\rm Supp}(\bc_\ell) |\notag \\
&=n- |\{ j: ~c_{1j}=c_{2j}=\cdots=c_{\ell j}=0,~1\leq j\leq n\}|.	
\end{align}
Notice that $\bc_i= \ba_i G$ for some $\ba_i\in \gf(q)^k$ and $c_{1j}=c_{2j}=\cdots=c_{\ell j}=0$ if and only if $$\ba_1 \bg_j^T=\ba_2 \bg_j^T=\cdots=\ba_\ell \bg_j^T=0,$$
i.e., $\bg_j\in V_B$, where
$$B:=\begin{bmatrix}
\ba_1\\
\ba_2\\
\vdots\\
\ba_\ell	
\end{bmatrix}.
 $$ 
Since $\C$ is projective, $\bg_1, \ldots, \bg_n$ are pairwise distinct.  
 By (\ref{eq:3}), we get 
 \begin{equation}\label{eq:4}
 \wt(\bc)=n-\left| \{\bg_1,\bg_2, \ldots, \bg_n \} \cap V_B \right|.	
 \end{equation} 
 Therefore, $\wt(\bc)=w$ if and only if $\left|\{\bg_1,\bg_2, \ldots, \bg_n \} \cap V_B\right|=n-w.$
 It is easily verified that $\bc$ runs over each codeword in $\C(q|q^\ell)$ once when $B$ runs over each $\ell \times k$ matrix in $\gf(q)^{\ell \times k}$ once. By (\ref{eq:4}), we obtain
 \begin{equation}\label{eq:5}
 	 A_w(\C(q |q^\ell ))=\left|\left \{ B \in \gf(q)^{\ell \times k}:\left|\{\bg_1, \ldots ,\bg_n\} \cap V_B \right|=n-w   \right \}\right|.
 \end{equation} The desired Result 3 follows.  
\end{proof}

The part of Theorem \ref{thm-fundam2} about the dimension and minimum distance of $\C(q|q^\ell)$ was implied in  \cite[Theorem 7]{DT21},  but the rest parts of Theorem \ref{thm-fundam2} look new and will be used to check if some designs are new or not. Equations \eqref{eq:4} and  (\ref{eq:5}) are useful to settle the weight distribution of certain lifted codes.

\begin{theorem}\label{thm-fundam3}
Let $n $ and $\ell$ be positive integers and let $\C$ be an $[n, k]_q$ linear code.  
Then 
$$
\C(q|q^\ell)^\perp = \C^\perp (q|q^\ell). 
$$
\end{theorem}

\begin{proof}
By Theorem \ref{thm-fundam2}, $\C(q|q^\ell)^\perp$ has dimension $n-k$. On the other hand, 
$\C^\perp$ has also dimension $n-k$. Consequently, any parity check matrix of  
$\C$ is also a generator matrix of $\C(q|q^\ell)^\perp$. The desired conclusion then follows.   
\end{proof} 

The following theorem provides a way to compute the weight enumerator of a lifted code $\C(q|q^\ell)$ 
\cite{HKM76}. 

\begin{theorem}\label{thm-HKM76} \cite{HKM76}
Let $\C$ be an $[n, k]_q$ linear code and let $\ell$ be a positive integer.  Then the lifted code $\C(q|q^\ell)$ has weight enumerator 
\begin{eqnarray}\label{eqn-HKM76}
A(\C(q|q^\ell))(z) = 1+\sum_{i=1}^n \sum_{j=1}^k N_{i}^{(j)} (q^\ell -1)(q^\ell -q) \cdots (q^\ell -q^{j-1})z^i, 
\end{eqnarray} 
where $N_{i}^{(j)}$ is the number of $(k-j)$-dimensional subspaces of $\gf(q)^k$ which contain exactly 
$n-i$ of the $n$ columns of a generator matrix $G$ of $\C$.   
\end{theorem} 

In theory, the formula in \eqref{eqn-HKM76} can be employed to compute the weight enumerator of not only 
a lifted code $\C(q|q^\ell)$ but also its original code $\C$.  In other words, it gives a way to compute the weight 
enumerator of every linear code over every finite field.  But the difficulty to use this formula lies in the computation of the coefficients $N_{i}^{(j)}$. In some cases, it is possible to use this formula to settle 
the weight enumerator of some lifted codes \cite{SLH24}.  

\section{The support $2$-designs of the lifted projective Reed-Muller codes} 

In this section, we will introduce the projective Reed-Muller codes and study the $2$-designs 
supported by the lifted projective Reed-Muller codes.  

Let $m \geq 2$ be an integer. 
 A point of the projective geometry $\PG(m-1, \gf(q))$ is given in homogeneous coordinates by $(x_1,x_2,\dots, x_{m})$ where all $x_i$ are in $\mathrm{GF}(q)$ and are not all zero. Every point of $\PG(m-1, \gf(q))$ has $q-1$ coordinate representations, as $(ax_1, ax_2,\dots,ax_{m})$ and $(x_1, x_2,\dots,x_{m})$ generate the same $1$-dimensional subspace of $\mathrm{GF}(q)^{m}$ for any nonzero $a\in \mathrm{GF}(q)$. 

Let $\gf(q)[x_1, x_2,\dots, x_m]$ denote the set of polynomials in $m$ indeterminates over $\gf(q)$, which is a linear space over $\gf(q)$. Let $A(q, m, h)$ denote the subspace of $\gf(q)[x_1, x_2,\dots, x_m]$ generated by all the homogeneous polynomials of degree $h$. Let $n=(q^m-1)/(q-1)$ and let $\{\mathbf{x}^1, \mathbf{x}^2, \ldots, \mathbf{x}^{n}\}$ be a set of projective points in $\mathrm{PG}(m-1,\gf(q))$. Then the \emph{$h$-th order projective Reed-Muller code}\index{projective Reed-Muller code} $\mathrm{PRM}(q, m, h)$ of length $n$ is defined by 
\begin{eqnarray*} 
\mathrm{PRM}(q,m,h)=\left \{\left (f(\mathbf{x}^1),f(\mathbf{x}^2), \dots, f(\mathbf{x}^n) \right ): f\in A(q, m, h) \right \}.
\end{eqnarray*} 
The code $\mathrm{PRM}(q,m,h)$ depends on the choice of the set  $\{\mathbf{x}^1,\mathbf{x}^2, \ldots, \mathbf{x}^{n}\}$ of coordinate representatives of the point set in $\mathrm{PG}(m-1,\gf(q))$, but is unique up to the monomial equivalence. The parameters of $\mathrm{PRM}(q,m,h)$ and $\mathrm{PRM}(q, m, h)^{\bot}$ are known and documented in the following theorems \cite{BR14,Lachaud,Sorensen}. 

\begin{theorem}\label{thm-PRMcode7} 
Let $m \geq 2$ and $1 \leq h \leq (m-1)(q-1)$. Then the linear code $\mathrm{PRM}(q,m,h)$ has length $n=(q^m-1)/(q-1)$ and minimum Hamming distance $(q-v)q^{m-2-u}$, where $h-1=u(q-1)+v$ and $0 \leq v <q-1$. Furthermore,  the dimension of the code is given as 
\begin{eqnarray*}\label{eqn-PRMcode7}
\dim(\mathrm{PRM}(q,m,h))=\sum_{t \equiv h \pmod{q-1} \atop 0 < t \leq h} \left(  \sum_{j=0}^m (-1)^j \binom{m}{j} \binom{t-jq+m-1}{t-jq} \right).   
\end{eqnarray*}
\end{theorem}  

\begin{theorem}\label{thm-prm8}
Let $m\geq 2$ and $1\leq h\leq (m-1)(q-1)$. If $h\not \equiv 0\pmod{q-1}$, then $$\mathrm{PRM}(q, m, h)^{\bot}=\mathrm{PRM}(q, m, (m-1)(q-1)-h).$$
\end{theorem}

By Theorem \ref{thm-PRMcode7} and definition, $\mathrm{PRM}(q,m,1)$ 
is monomially-equavalent to the Simplex code. 
It then follows from Theorem \ref{thm-prm8} that 
$\mathrm{PRM}(q,m,(m-1)(q-1)-1)$ 
is monomially-equavalent to the Hamming code.  
Thus, the Hamming codes and Simplex codes are special projective Reed-Muller codes. 

It was pointed out in \cite{BM01,Sorensen} that the code $\mathrm{PRM}(q,m,h)$ is not cyclic in general, but is equivalent to a cyclic code if $\gcd(m, q-1)=1$ or $h \equiv 0 \pmod{q-1}$.  It was proved in \cite{SDW24} that 
every projective Reed-Muller code  $\mathrm{PRM}(q,m,h)$ is a constacyclic code.

For a linear code $\C$, define
$$
\pi \GAut(\C)=\{P: DP\gamma \in \GAut(\C)\} 
$$
and 
$$
\pi \MAut(\C)=\{P: DP \in \MAut(\C)\}.  
$$
By definition 
\begin{eqnarray}
\pi \MAut(\C) \subseteq \pi \GAut(\C). 
\end{eqnarray}

The following result was proved in \cite{Berger02} and will be very useful for studying the designs supported by the projective Reed-Muller codes, the Hamming and Simplex codes and the lifted codes of these codes. 

\begin{lemma}\label{lem-berger}
Choose the point set $\{\mathbf{x}^1, \mathbf{x}^2, \ldots, \mathbf{x}^{n}\}$ of $\mathrm{PG}(m-1,\gf(q))$ 
to index the coordinates of the codewords in $\PRM(q,m,h)$.    Then 
$$
\PGL_m(\gf(q)) \subseteq \pi \MAut(\PRM(q,m,h)) \subseteq \pi \GAut(\PRM(q,m,h)), 
$$
where $\PGL_m(\gf(q)) $ denotes the projective general linear group \cite[Chapter 1]{DTbook22}.
\end{lemma}

\begin{theorem}\label{thm-designresult191} 
Let $m \geq 2$ and $\ell \geq 1$.   
For each positive integer $i$ with $A_i(\PRM(q,m,h)(q|q^\ell))>0$, the incidence structure 
$\bD_i(\PRM(q,m,h)(q|q^\ell))$ is a $2$-$(n,i,\lambda_i)$ design for some $\lambda_i$.   
\end{theorem}

\begin{proof}
It is known that $\PGL_m(\gf(q))$ is doubly transitive \cite[Chapter 1]{DTbook22}.  The desired conclusion then follows from 
Lemma \ref{lem-berger} and Theorems \ref{thm-designCodeAutm} and \ref{thm-main}.   
\end{proof}

\begin{open}\label{open-liftedPRMcodes} 
Determine the weight enumerator of the lifted codes $\PRM(q,m,h)(q|q^\ell)$. 
\end{open} 

The weight distribution of $\mathrm{PRM}(q,m,2)$ was settled in \cite{Lisx19}. 
Therefore, the weight distribution of $\mathrm{PRM}(q,m,(m-1)(q-1)-2)$ is known. 
But the weight distribution of $\mathrm{PRM}(q,m,h)$ is open for $2<h < (m-1)(q-1)-2$.  
The following open problem is harder. 

\begin{open} 
Determine the values of $i$ and $\lambda_i$ for the $2$-$(n,i, \lambda_i)$ designs in Theorem \ref{thm-designresult191}. 
\end{open}

\section{The support $2$-designs of the lifted Simplex and Hamming codes}

Throughout this section, let $n=(q^{m}-1)/(q-1)$ with $m \geq 2$ being an integer. 

\subsection{The support designs of Hamming and Simplex codes over finite fields} 

There are several different constructions of the Hamming and Simplex codes. The first one is a trace construction  
of  the Hamming and Simplex codes. Let $\alpha$ be a generator of $\gf(q^m)^*$.  
Define 
$$ 
\Delta_i:=\alpha^i \gf(q)^*=\{\alpha^i a: a \in \gf(q)^*\}
$$  
for all $i$ with $0 \leq i \leq n-1$. Clearly, the set $\{\Delta_i: 0 \leq i \leq n-1\}$ forms a partition of $\gf(q^m)^*$. 
Let $i$ and $j$ be a pair of distinct elements in the set $\{0, 1, \ldots, n-1\}$. Then any 
$a \in \Delta_i$ and $b \in \Delta_j$ must be linearly independent over $\gf(q)$. 

Let $b_i \in \Delta_i$ for each $i$ with $0 \leq i \leq n-1 $.  Define 
\begin{eqnarray}
\cS_{(q,m)}(b_0, \ldots, b_{n-1})=\left\{\left(\tr_{q^m/q}(ab_i)\right)_{i=0}^{n-1}: a \in \gf(q^m) \right\}. 
\end{eqnarray} 
Then the set $\cS_{(q,m)}(b_0, \ldots, b_{n-1})$ defined above  is a Simplex code over $\gf(q)$ with parameters 
$[(q^m-1)/(q-1), m, q^{m-1}]$ and weight enumerator $1+(q^m-1)z^{q^{m-1}}$. By definition, different choices of 
the vector $(b_0, \ldots, b_{n-1})$  in the set $ \Delta_0 \times \cdots  \times \Delta_{n-1}$ 
result in monomially-equivalent Simplex codes.  It is well known that monomially-equivalent codes have the same parameters and same weight enumerator. Therefore, up to monomial equivalence, Simplex codes are unique, and are 
denoted by $\cS_{(q,m)}$ in this paper.   

The dual code of any Simplex code $\cS_{(q,m)}(b_0, \ldots, b_{n-1})$ is referred to as a Hamming code, denoted by 
$\cH_{(q,m)}(b_0, \ldots, b_{n-1})$. Thus, all the Hamming codes $\cH_{(q,m)}(b_0, \ldots, b_{n-1})$ are  monomially-equivalent and unique up to monomial equivalence. Hence, we 
denote them  by $\cH_{(q,m)}$. It is well known that $\cH_{(q,m)}$ has 
parameters $[(q^m-1)/(q-1), (q^m-1)/(q-1)-m, 3]$. 

The second construction of the Hamming and Simplex codes is by matrix.   
A parity check matrix $H_{(q,m)}$ of the \emph{Hamming code\index{Hamming code}} $\cH_{(q,m)}$ over $\gf(q)$ 
 is defined by choosing its columns as the nonzero vectors from 
all the one-dimensional subspaces of $\gf(q)^m$. In terms of finite geometry, the columns of $H_{(q,m)}$ are the 
points of the projective geometry $\PG(m-1, \gf(q))$ \cite[Section 1.8]{DTbook22}. Hence, $\cH_{(q,m)}$ has length $n=(q^m-1)/(q-1)$ and dimension 
$n-m$. By definition, no two columns of $H_{(q,m)}$ are linearly dependent over $\gf(q)$. The minimum weight of $\cH_{(q,m)}$ 
is at least 3. Adding two nonzero vectors from two different one-dimensional subspaces gives a nonzero vector
from a third one-dimensional space. Therefore,  $\cH_{(q,m)}$ has minimum weight 3. 
It is also well known that any $[(q^m-1)/(q-1), (q^m-1)/(q-1)-m, 3]$ code over $\gf(q)$ is monomially equivalent to the Hamming code $\cH_{(q,m)}$ \cite[Theorem 1.8.2]{HP03}. 

The third one is a contacyclic code construction of the Hamming and Simplex codes \cite{SDW24}.  The fourth one is the following.  

\begin{theorem}\label{thm-july191}
Each Simplex code is a projective Reed-Muller code $\mathrm{PRM}(q,m,1)$ and each Hamming code 
is a projective Reed-Muller code $\mathrm{PRM}(q,m,(m-1)(q-1)-1)$. 
\end{theorem}

\begin{proof} 
It is  well known that any $[(q^m-1)/(q-1), (q^m-1)/(q-1)-m, 3]$ linear code over $\gf(q)$ is monomially equivalent to the Hamming code $\cH_{(q,m)}$ \cite[Theorem 1.8.2]{HP03}.  
By  Theorem \ref{thm-PRMcode7},   $\mathrm{PRM}(q,m,(m-1)(q-1)-1)$ has the same parameters as 
a Hamming code $\cH_{(q,m)}$ and thus they are monomially-equivalent.  Note all projective RM codes 
are monomially-equivalent and all Hamming codes are monomially-equivalent. The desired second conclusion then 
follows.   
The desired first conclusion 
then follows from Theorem \ref{thm-prm8}. 
\end{proof}

The weight distribution of  $\cH_{(q,m)}$  is given in the following lemma 
\cite{DL17}. 

\begin{lemma}\label{lem-HCwt}  \cite[p. 2418]{DL17} 
The weight distribution of $\cH_{(q,m)}$ is given by
\begin{eqnarray*}
q^m A_{k}(\cH_{(q,m)})= 
& \sum_{\substack {0 \le i \le \frac{q^{m-1}-1}{q-1} \\0 \le j \le q^{m-1} \\ i+j=k}}\left[\binom{\frac{q^{m-1}-1}{q-1}}{i}
 \binom{q^{m-1}}{j}\Big((q-1)^k+(-1)^j(q-1)^i(q^m-1)\Big)\right]
\end{eqnarray*}
for $0 \leq k \leq (q^m-1)/(q-1)$.
\end{lemma} 

As a corollary of Theorems \ref{thm-designresult191} and  \ref{thm-july191}, we have the following. 

\begin{theorem}\label{thm-designresult192}
For each positive integer $i$ with $A_i(\cH_{(q,m)})>0$, the incidence structure 
$\bD_i(\cH_{(q,m)})$ is a $2$-$(n,i,\lambda_i)$ design for some $\lambda_i$.  
Furthermore, $\bD_{q^{m-1}}(\cS_{(q,m)})$ is a $2$-$(n,q^{m-1}, (q-1)q^{m-2})$ design. 
\end{theorem}

Theorem \ref{thm-designresult192} documents all the support designs of the Hamming codes, which was 
obtained by investigating the automorphism group of the Hamming codes.  
With the help of the Assmus-Mattson theorem it was proved in \cite[Theorem 10.23]{DTbook22} that $\bD_{q^{m-1}}(\cS_{(q,m)})$ is a $2$-$(n, q^{m-1}, (q-1)q^{m-2})$ design and the following hold:
\begin{itemize}
\item $\bD_{3}(\cH_{(q,m)})$ is a $2$-$(n, 3, q-1)$ design for each prime power $q \geq 2$.  
 \item $\bD_{4}(\cH_{(3,m)})$ is a $2$-$((3^m-1)/2, 4, \lambda_m)$ design for some $\lambda_m$. 
\end{itemize}
These are the known $2$-designs supported by the Hamming and Simplex codes obtained via the Assmus-Mattson theorem.  Hence, the automorphism group approach is much more powerful than the Assmus-Matsson theorem approach when they are used for studying designs supported by the Hamming codes.

\subsection{The support designs of lifted Simplex and Hamming codes} 

Let $\ell$ be a positive integer with $\ell \leq m$ and let $\cS_{(q,m)}(q|q^\ell)$ denote the lifted code over $\gf(q^\ell)$ of 
the Simplex code $\cS_{(q,m)}$.  

\begin{theorem}\label{thm-paramSimplex}
Let $m \geq 2$. Then the following hold. 
\begin{itemize}
\item The lifted code $\cS_{(q,m)}(q|q^\ell)$ has  parameters $[n, m, q^{m-1}]_{q^\ell}$.  
\item The dual code $\cS_{(q,m)}(q|q^\ell)^{\perp}$ has parameters $[n, n-m, 3]_{q^\ell}$. 
\end{itemize} 
\end{theorem}

\begin{proof}
It follows from Theorem \ref{thm-fundam2} that $\cS_{(q,m)}(q|q^\ell)$ has the same parameters as the Simplex code $\cS_{(q,m)}$.  
It follows from Theorem \ref{thm-fundam3} that $\cS_{(q,m)}(q|q^\ell)^{\perp}$ has the same parameters as the Hamming code $\cH_{(q,m)}$.  
The desired conclusions then follow from the parameters of the Simplex code and Hamming code.  
\end{proof}

The weight enumerator of a lifted Simplex code $\cS_{(q,m)}(q|q^\ell)$ was determined in \cite{SLH24} 
with the help of Theorem \ref{thm-HKM76}.  Below we present it in a slightly different way and give it a direct 
proof. 

\begin{theorem}\label{thm-liftedSimplexWD} \cite{SLH24}
Let $m \geq 2$ and $1 \leq \ell \leq m$. Then the lifted code $\cS_{(q,m)}(q|q^\ell)$ has weight enumerator 
\begin{eqnarray}
1+ \sum_{r=1}^{\ell} \left( \prod_{j=1}^r \frac{(q^\ell -q^{j-1})(q^m -q^{j-1})}{q^{j-1}(q^j-1)} \right) 
z^{q^{m-r}(q^{r}-1)/(q-1)}.  
\end{eqnarray}
\end{theorem}

\begin{proof}
We follow the notation of the proof of Theorem \ref{thm-fundam3}. 
Let $G=[\bg_1^T \ldots \bg_n^T]$ be a generator matrix of the Simplex code $\cS_{(q,m)}$. By definition the set 
$\{ \bg_1, \ldots, \bg_n\}$ is a point set of 
$\PG(m-1, \gf(q))$.  For any $B\in \gf(q)^{\ell \times m}$ with rank $1\leq r\leq \ell$, it is easily verified that $$\left|\{\bg_1, \bg_2,\ldots, \bg_n  \} \cap V_B \right|=\frac{q^{m-r}-1}{q-1}.$$
By (\ref{eq:4}), the nonzero weights of $\C(q| q^\ell)$ are $\{w_r:=q^{m-r}(q^r-1)/(q-1):~1\leq r\leq \ell \}$. Note that $1 \leq r \leq \ell$ and the number of $\ell \times m$ matrices over $\gf(q)$ with rank $r$ is equal to 
$$
 \prod_{j=1}^r \frac{(q^\ell -q^{j-1})(q^m -q^{j-1})}{q^{j-1}(q^j-1)}.
$$
The total number of codewords with weight $w_r$ in the lifted code 
$\cS_{(q, m)}(q|q^\ell)$ directly follows from (\ref{eq:5}). This completes the proof.  
\end{proof}

As a corollary of Theorem \ref{thm-liftedSimplexWD}, we have the following family of projective two-weight codes, which have the same parameters as the two-weight codes in Example RT1 in \cite{CK1986}.

\begin{corollary}\label{cor-2wtcode}
Let $m \geq 2$. Then the lifted code $\cS_{(q,m)}(q|q^2)$ has weight enumerator 
$$
1+(q+1)(q^m-1)z^{q^{m-1}}+(q^m-q)(q^m-1)z^{q^{m-1}+q^{m-2}}. 
$$
\end{corollary}

As a corollary of Theorem \ref{thm-liftedSimplexWD}, we have the following family of projective three-weight codes.

\begin{corollary}\label{cor-3wtcode}
Let $m \geq 3$. Then the lifted code $\cS_{(q,m)}(q|q^3)$ has weight enumerator 
\begin{eqnarray*}
1+(q^2+q+1)(q^m-1)z^{q^{m-1}}+(q^2+q+1)(q^m-1)(q^m-q)z^{q^{m-1}+q^{m-2}} + \\
(q^m-1)(q^m-q)(q^m-q^2)z^{q^{m-1}+q^{m-2}+q^{m-3}}. 
\end{eqnarray*}
\end{corollary}

\begin{theorem}\label{thm-liftedSimplexCodeDesigns2}
Let $m \geq 2$ be an integer and $1  \leq r \leq \ell \leq m$.  Then incidence structure 
$$\bD_{q^{m-r}(q^r-1)/(q-1)}(\cS_{(q,m)}(q|q^\ell) )$$ 
is a 
$2$-$(n, q^{m-r}(q^r-1)/(q-1), \lambda_r)$ design for some positive integer $\lambda_r$.   
\end{theorem}

\begin{proof}
By Lemma \ref{lem-berger} and Theorem \ref{thm-july191},  $\pi\MAut(\cS_{q,m)})$ is doubly transitive. 
The $2$-design property follows from Theorems \ref{thm-main} and \ref{thm-liftedSimplexWD}.   
\end{proof} 

We have the following remarks about the designs in Theorem \ref{thm-liftedSimplexCodeDesigns2}. 
\begin{itemize}
\item It follows from Theorem \ref{thm-fundam2} that 
$$
 \cB_{q^{m-1}}(\cS_{(q,m)}(q|q^\ell) ) =  \cB_{q^{m-1}}(\cS_{(q,m)} ). 
$$
Thus, the design $\bD_{q^{m-1}}(\cS_{(q,m)}(q|q^\ell) )$ is the same as the design $\bD_{q^{m-1}}(\cS_{(q,m)})$ and is not new.  
\item The designs $\bD_{q^{m-r}(q^r-1)/(q-1)}(\cS_{(q,m)}(q|q^\ell) )$ for $2 \leq r \leq \ell$ look new.  
However, the size of the block set $\cB_{q^{m-r}(q^r-1)/(q-1)}(\cS_{(q,m)}(q|q^\ell) )$ for $2 \leq r \leq \ell$ and thus the corresponding $\lambda_r$ are hard to determine.  
\end{itemize} 

\begin{example} 
Consider the case $(q, m, \ell)=(2,4,2)$. 
Then the lifted code $\cS_{(2,4)}(2|2^2)$ has parameters $[15, 4, 8]_{4}$ and weight enumerator 
$1+45z^8+210z^{12}.$ Furthermore,  the following hold:
\begin{itemize}
\item $\bD_{8}(\cS_{(2,4)}(2|2^2))$ is a $2$-$(15, 8, 4)$ design.
\item $\bD_{12}(\cS_{(2,4)}(2|2^2))$ is a $2$-$(15, 12, 22)$ design.  
\end{itemize}
\end{example} 

\begin{example}\label{exam-simplex} 
Consider the case $(q, m,\ell)=(3, 4,3)$. Then the lifted code $\cS_{(3,4)}(3|3^3)$ has parameters $[40, 4, 27]_{27}$ and weight enumerator 
$1+1040z^{27}+81120z^{36}+449280z^{39}.$ Furthermore,  the following hold:
\begin{itemize}
\item $\bD_{27}(\cS_{(3,4)}(3|3^3))$ is a $2$-$(40, 27, 18)$ design.
\item $\bD_{36}(\cS_{(3,4)}(3|3^3))$ is a $2$-$(40, 36, 105)$ design. 
\item $\bD_{39}(\cS_{(3,4)}(3|3^3))$ is a $2$-$(40, 39, 38)$ design (i.e.,  a complete design). 
\end{itemize}
\end{example}

\begin{open}
Determine the values $\lambda_r$ for the $2$-$(n, q^{m-r}(q^r-1)/(q-1), \lambda_r)$  designs for $r \geq 2$  
in Theorem \ref{thm-liftedSimplexCodeDesigns2}. 
\end{open}

\begin{theorem}\label{thm-liftedHammingWD}
Let $m \geq 2$ and $1 \leq \ell \leq m$. Then the lifted code $\cH_{(q,m)}(q|q^\ell)$ has parameters 
$[n, n-m, 3]_{q^\ell}$ and weight enumerator 
$$
q^{-\ell m} (1+(q^\ell-1)z )^n \ A \left(   \frac{1-z}{1+(q^\ell -1)z}\right), 
$$
where 
\begin{eqnarray*}
A(z)=1+ \sum_{r=1}^{\ell} \left( \prod_{j=1}^r \frac{(q^\ell -q^{j-1})(q^m -q^{j-1})}{q^{j-1}(q^j-1)} \right) z^{q^{m-r}(q^r-1)/(q-1)}.  
\end{eqnarray*}
\end{theorem}

\begin{proof}
The desired conclusions follow from Theorems \ref{thm-fundam3}, \ref{thm-liftedSimplexWD}, \ref{thm-fundam2} and the MacWilliams Identity. 
\end{proof}

\begin{theorem}\label{thm-liftedHammingCodeDesigns}
Let $m \geq 2$ be an integer and $1 \leq \ell \leq m$.  
If $A_i(\cH_{(q,m)}(q|q^\ell)) >0$ for some $3 \leq i \leq n$,  then the incidence structure 
$\bD_{i}(\cH_{(q,m)}(q|q^\ell) )$ 
is a 
$2$-$(n, i, \lambda_i)$ design for some positive integer $\lambda_i$.   
\end{theorem}

\begin{proof}
By Lemma \ref{lem-berger} and Theorem \ref{thm-july191},  $\pi\MAut(\cH_{q,m)})$ is doubly transitive. 
The $2$-design property follows from Theorems \ref{thm-main}.   
\end{proof} 

We have the following remarks about the designs in Theorem \ref{thm-liftedHammingCodeDesigns}. 
\begin{itemize}
\item Although the weight enumerator of the lifted code $\cH_{(q,m)}(q|q^\ell)$ was explicitly given in 
Theorem \ref{thm-liftedHammingWD},  it is not easy to see if $A_i(\cH_{(q,m)}(q|q^\ell)) >0$ for a specific $i \geq 4$ 
and it is much harder to determine the parameters of a $2$-design $\bD_{i}(\cH_{(q,m)}(q|q^\ell) )$ in 
Theorem \ref{thm-liftedHammingCodeDesigns}.   
\item When  $A_i(\cH_{(q,m)}) >0$,  the designs $\bD_{i}(\cH_{(q,m)})$ and $\bD_{i}(\cH_{(q,m)}(q|q^\ell) )$ may have different values $\lambda_i$ for $i>3$.  When  $A_i(\cH_{(q,m)}) =0$,  it may happen that $A_i(\cH_{(q,m)}(q|q^\ell)) >0$. Hence, certain new $2$-designs are produced in Theorem \ref{thm-liftedHammingCodeDesigns}.  
The following example illustrates these facts. 
\end{itemize}

\begin{example} 
Consider the case $(q, m,\ell)=(2, 4,2)$.  Then the Hamming code $\cH_{(2,4)}$ has parameters $[15, 11,3]_2$ and weight enumerator 
$$ 
1+ 35z^3 + 105z^4 + 168z^5  + 280z^6 +  435z^7 + 435z^8 + 9, 280z^9 + 
168z^{10} + 105z^{11} + 35z^{12} + z^{15} . 
$$ 
In addition, the following hold: 
\begin{itemize}
\item $\bD_{3}(\cH_{(2,4)})$ is a $2$-$(15, 3, 1)$ design.
\item $\bD_{4}(\cH_{(2,4)})$ is a $2$-$(15, 4, 6)$ design.
\item $\bD_{5}(\cH_{(2,4)})$ is a $2$-$(15, 5, 16)$ design.
\item $\bD_{6}(\cH_{(2,4)})$ is a $2$-$(15, 6, 40)$ design.
\item $\bD_{7}(\cH_{(2,4)})$ is a $2$-$(15, 7, 87)$ design.
\item $\bD_{8}(\cH_{(2,4)})$ is a $2$-$(15, 8, 116)$ design. 
\item $\bD_{9}(\cH_{(2,4)})$ is a $2$-$(15, 9, 96)$ design.
\item $\bD_{10}(\cH_{(2,4)})$ is a $2$-$(15, 10, 72)$ design.
\item $\bD_{11}(\cH_{(2,4)})$ is a $2$-$(15, 11, 55)$ design. 
\item $\bD_{12}(\cH_{(2,4)})$ is a $2$-$(15, 12, 22)$ design.
\item $\bD_{15}(\cH_{(2,4)})$ is a $2$-$(15, 15, 1)$ design (a complete design).
\end{itemize}
The lifted code $\cH_{(2,4)}(2|2^2)$ has parameters $[15, 11, 3]_{4}$ and weight enumerator
\begin{eqnarray*}
1 +  105z^3 +  315z^4 + 2394z^5 + 15750z^6 +  54855z^7 + 160695z^8 +  391020z^9  +  \\
 688212z^{10} +  949095z^{11} + 937965z^{12} + 659610z^{13} + 
277830z^{14}  + 56457z^{15}. 
\end{eqnarray*} 
Furthermore,  the following hold:
\begin{itemize}
\item $\bD_{3}(\cH_{(2,4)}(2|2^2))$ is a $2$-$(15, 3, 1)$ design (the same as $\bD_{3}(\cH_{(2,4)})$).  
\item $\bD_{4}(\cH_{(2,4)}(2|2^2))$  is a $2$-$(15, 4, 6)$ design (the same as $\bD_{4}(\cH_{(2,4)})$).  
\item $\bD_{5}(\cH_{(2,4)}(2|2^2))$ is a $2$-$(15, 5, 46)$ design (different from $\bD_{5}(\cH_{(2,4)})$).  
\item $\bD_{6}(\cH_{(2,4)}(2|2^2))$ is a $2$-$(15, 6, 355)$ design (different from $\bD_{6}(\cH_{(2,4)})$).  
\item $\bD_{7}(\cH_{(2,4)}(2|2^2))$ is a $2$-$(15, 7, 1095)$ design (different from $\bD_{7}(\cH_{(2,4)})$).  
\item $\bD_{8}(\cH_{(2,4)}(2|2^2))$ is a $2$-$(15, 8, 1684)$ design (different from $\bD_{8}(\cH_{(2,4)})$).  
\item  $\bD_{i}(\cH_{(2,4)}(2|2^2))$ is a $2$-$(15, i, \binom{15}{i})$ design (a complete design) for each $i$ with $9 \leq i \leq 15$.  
\end{itemize}
This example shows some differences and similarities of the designs supported by the Hamming code 
$\cH_{(2,4)}$ and 
the lifted Hamming code $\cH_{(2,4)}(2|2^2)$.  
\end{example}

\begin{open}
Determine the parameters of the $2$-designs 
in Theorem \ref{thm-liftedHammingCodeDesigns}. 
\end{open} 

\section{Lifted Reed-Muller codes and their support designs} 

The reader is referred to \cite[Chappter 5]{DTbook22} or other books on coding theory for a 
description of the binary Reed-Muller codes $\RM_2(r,m)$ of order $r$.  The parameters of 
these Reed-Muller codes $\RM_2(r,m)$ are known and the weight distribution of $\RM_2(r,m)$ 
is known for certain $r$. It is known that the codes $\RM_2(r,m)$ support $3$-designs. However, 
 the generalised Reed-Muller codes $\RM_q(r,m)$ for $q>3$ support only $2$-designs but do not support $3$-designs \cite[Chapters 5 and 6]{DTbook22}.  The objectives of this section are to show that 
many infinite families of $3$-designs are supported by the lifted codes $\RM_2(r,m)(2|2^\ell)$ and settle the weight distributions of the lifted codes $\RM_2(1,m)(2|2^\ell)$ and $\RM_2(m-2,m)(2|2^\ell)$. 

The following theorem follows from Theorems  \ref{thm-fundam2} and \ref{thm-fundam3} and the basic results about the Reed-Muller codes $\RM_2(r,m)$. 

\begin{theorem}
Let $m \geq 3$ and $0 \leq r <m$. Then the lifted code $\RM_2(r,m)(2|2^\ell)$ has parameters 
$$
\left[ 2^m, \ \sum_{i=0}^r \binom{m}{i},  \ 2^{m-r}     \right]_{2^\ell}.
$$
Furthermore, $\RM_2(r,m)(2|2^\ell)^\perp$ has parameters 
$$
\left[ 2^m, \ 2^m-\sum_{i=0}^r \binom{m}{i},  \ 2^{r+1}     \right]_{2^\ell} 
$$
and 
$$
\RM_2(r,m)(2|2^\ell)^\perp=\RM_2(m-1-r,m)(2|2^\ell). 
$$
\end{theorem} 

The next theorem shows that many infinite families of $3$-designs are supported by the lifted 
codes of the binary Reed-Muller codes $\RM_2(r,m)$ of order $r$ for $1 \leq r<m$.  

\begin{theorem}\label{thm-last}
Let $m \geq 3$ and $\ell \geq 1$.  For each $i \geq 3$ with $A_i(\RM_2(r,m)(2|2^\ell))>0$, the incidence structure 
 $\bD_i(\RM_2(r,m)(2|2^\ell))$ is a $3$-$(2^m, i, \lambda_i)$ design for some positive integer $\lambda_i$. 
\end{theorem}

\begin{proof}
 Notice that 
$$
\Aut(\RM_2(r,m))=\MAut(\RM_2(r,m))=\PAut(\RM_2(r,m))=\GA_m(\gf(2)), 
$$
which is the general affine group and triply transitive acting on $\gf(2)^m$ \cite{DTbook22}. It then follows from Theorem 
\ref{thm-main} that $\bD_i(\RM_2(r,m)(2|2^\ell))$ is a $3$-design for each $i \geq 3$ with $A_i(\RM_2(r,m)(2|2^\ell))>0$.  
\end{proof}

\begin{open}
Determine the values $(i, \lambda_i)$ for the  $3$-$(2^m, i, \lambda_i)$ designs in Theorem \ref{thm-last}.  
\end{open}

It is well known that $\RM_2(1,m)$ has weight enumerator $1+(2^{m+1}-2)z^{2^{m-1}}+z^{2^m}$. 
The following theorem gives the weight distribution of the lifted code $\RM_2(1,m)(2|2^\ell)$.  

\begin{theorem}\label{thm-july241wtd}
Let $m \geq 3$ and $1 \leq \ell \leq m$. Then the weight distribution of $\RM_2(1,m)(2|2^\ell)$ is the following:
\begin{eqnarray*}
A_0 &=& 1, \\
A_{2^m-2^{m-h}} &=& 2^h \prod_{j=1}^h \frac{(2^\ell -2^{j-1})(2^m-2^{j-1})}{2^{j-1}(2^j-1)} \mbox{ for each } 1 \leq h \leq \ell, \\
A_{2^m} &=& 2^{\ell(m+1)}-1-\sum_{h=1}^\ell A_{2^m-2^{m-h}}, \\
A_i &=& 0 \ \mbox{ for other } \ i. 
\end{eqnarray*}
\end{theorem} 

\begin{proof} 
We follow the notation of the proof of Theorem \ref{thm-liftedSimplexWD}. 
Let $n=2^m$.  Let $\{ \bg_1, \ldots, \bg_n\}=\gf(2)^{m \times 1}$. By definition,  
\begin{eqnarray*}
G(1,m)=\left[ 
\begin{array}{lll}
\bg_1 & \cdots & \bg_n \\
1       &  \cdots & 1 
\end{array}
\right] =\left[ \ba_1^T \cdots \ba_n^T\right] 
\end{eqnarray*} 
is a generator matrix of $\RM_2(1,m)$.

For any $B=[B_1 B_2]\in \gf(q)^{\ell \times (m+1)}$, where $B_1\in \gf(q)^{\ell \times m}$ and $B_2\in \gf(q)^{\ell \times 1}$. It is easily seen that  
\begin{eqnarray*}
B 
\left[ 
\begin{array}{l}
x_1 \\
\vdots \\
x_m \\
1 
\end{array}
\right] 
=\bzero^T 
\mbox{ if and only if } 
B_1 
\left[ 
\begin{array}{l}
x_1 \\
\vdots \\
x_m 
\end{array}
\right] 
=B_2.
\end{eqnarray*}
It follows that 
$$ 
\left|\{\ba_1, \ldots, \ba_n  \} \cap V_B  \right|=\begin{cases}
	0 & \mbox{ if } \rank(B_1)\neq \rank(B),\\
	2^{m-h} & \mbox{ if } \rank(B_1)=\rank(B)=h,
\end{cases}
 $$
where $1 \leq h \leq \ell$. By (\ref{eq:4}), the nonzero weights of $\C(q| q^\ell)$ are $$\{w_h:=2^m-2^{m-h}:~1\leq h \leq \ell \} \cup \{2^m\}.$$

We now determine $A_{2^m-2^{m-h}}$ for $1 \leq h \leq \ell$.  Let $\rank(B_1)=h$. It is known that 
\begin{eqnarray}\label{eqn-july2412}
|\{B_1 \in \gf(2)^{\ell \times m}: \rank(B_1)=h \}| = \prod_{j=1}^h \frac{(2^\ell -2^{j-1})(2^m-2^{j-1})}{2^{j-1}(2^j-1)}. 
\end{eqnarray}
Note that 
\begin{eqnarray}\label{eqn-july2413}
|\{B_1 \bx : \bx \in \gf(2)^{m \times 1} \}|=2^h. 
\end{eqnarray} 
Combining \eqref{eq:5}, \eqref{eqn-july2412} and \eqref{eqn-july2413} yields the desired formula for $A_{2^m-2^{m-h}}$. This completes the proof. 
\end{proof}

Combining Theorems \ref{thm-last} and and \ref{thm-july241wtd} yields the following. 

\begin{corollary}\label{cor-july24design3} 
Let $m \geq 3$ and $1 \leq h \leq \ell \leq m$. Then $\bD_{2^m-2^{m-h}}(\RM_2(1,m)(2|2^\ell))$ is a $3$-$(2^m,  2^m-2^{m-h}, \lambda_h)$ design for some positive integer $\lambda_h$.  
\end{corollary}  

\begin{open}
Determine the value $\lambda_h$ with $2 \leq h \leq \ell$ for the $3$-$(2^m,  2^m-2^{m-h}, \lambda_h)$ design 
in Corollary \ref{cor-july24design3}. 
\end{open} 

We have the following remarks about the $3$-designs in Corollary \ref{cor-july24design3}. 
\begin{itemize}
\item It follows from Theorem \ref{thm-fundam2} that $\bD_{2^{m-1}}(\RM_2(1,m)(2|2^\ell))$ is 
the same as $\bD_{2^{m-1}}(\RM_2(1,m))$ and is not new. 
\item The other designs $\bD_{2^m-2^{m-h}}(\RM_2(1,m)(2|2^\ell))$ with $2 \leq h \leq \ell$ look new. 
\end{itemize}

As a corollary of Theorem \ref{thm-july241wtd},  we have the following. 

\begin{corollary}\label{cor-3wtcoderm} 
Let $m \geq 3$. Then the lifted code $\RM_2(1,m)(2|2^2)$ has weight enumerator 
$$
1+3(2^{m+1}-2)z^{2^{m-1}} + 4(2^m-1)(2^m-2)z^{3\times 2^{m-2}} + 3(2^{m+1}-1)z^{2^{m}}. 
$$
\end{corollary}

\begin{conj}\label{conj-august31}
Let $m \geq 3$.  Then 
$\bD_{3\times 2^{m-2}}(\RM_2(1,m)(2|2^2))$ is a $3$-$(2^m,  3\times 2^{m-2}, \lambda_m)$ design with 
$$
\lambda_m = \frac{2\binom{3\times 2^{m-2}}{3} (2^m-1)(2^m-2)}{3\binom{2^m}{3}}
$$
\end{conj} 

\begin{example}
The lifted code $\RM_2(1,4)(2|2^2)$ has parameters $[16,5,8]_4$ and weight enumerator 
$ 
1+ 90z^8 + 840z^{12} +93z^{16}. 
$ 
Furthermore, $\bD_{12}(\RM_2(1,4)(2|2^2))$ is a $3$-$(16,  12, 55)$ design.  
\end{example}

\begin{open} 
Determine the weight enumerators of the lifted codes  $\RM_2(r,m)(2|2^\ell)$ for $2 \leq r \leq m-3$ and $2 \leq \ell \leq m$.   
\end{open}

\begin{theorem}\label{thm-liftedRMWD2}
Let $m \geq 3$ and $1 \leq \ell \leq m$. Then the lifted code $\RM_2(m-2,m)(2|2^\ell)$ has  
weight enumerator 
$$
2^{-\ell (m+1)} (1+(2^\ell-1)z )^{2^m} \ A \left(   \frac{1-z}{1+(2^\ell -1)z}\right), 
$$
where 
\begin{eqnarray*}
A(z)=1+\sum_{i=1}^{2^m} A_i z^i 
\end{eqnarray*}
and these $A_i$ were given in Theorem \ref{thm-july241wtd}. 
\end{theorem}

\begin{proof}
The desired conclusion follow from Theorems \ref{thm-fundam3}, \ref{thm-july241wtd}, \ref{thm-fundam2} and the MacWilliams Identity. 
\end{proof}

It is known that $D_i(\RM_q(r,m))$ is only a $2$-design but not a $3$-design 
for $q>2$ \cite[Chapter 6]{DTbook22}. Hence, the $3$-designs $\bD_i(\RM_2(r,m)(2|2^\ell))$  in Theorem \ref{thm-last} are  
valuable and interesting. 

\section{Summary and concluding remarks}\label{sec-last}

The main contributions of this paper are summarized as follows. 
\begin{itemize}
\item Certain fundamental results for lifted linear codes were proved in Theorems \ref{thm-fundam2} and \ref{thm-fundam3}. 
\item The support $2$-designs of  the lifted projective Reed-Muller codes were studied  in Theorem 
\ref{thm-designresult191}.  New infinite families of $2$-designs were obtained by studying  the lifted projective Reed-Muller codes. 
\item The support $2$-designs of  the lifted Simplex codes $\cS_{(q,m)}(q|q^\ell)$ were characterised in  Theorem  \ref{thm-liftedSimplexCodeDesigns2}.  New infinite families of $2$-designs were obtained by studying the lifted Simplex codes. 
\item The support $2$-designs of  the lifted Hamming codes $\cH_{(q,m)}(q|q^\ell)$ were characterised 
in Theorem  \ref{thm-liftedHammingCodeDesigns}.  New infinite families of $2$-designs were obtained by studying the lifted Hamming codes. 
\item An infinite family of three-weight projective codes over $\gf(4)$ was obtained in Corollary  \ref{cor-3wtcoderm}.  
\item The weight distributions of the lifted codes  $\RM_2(1,m)(2|2^\ell)$ were settled in Theorem \ref{thm-july241wtd}. 
\item The weight distributions of the lifted codes  $\RM_2(m-2,m)(2|2^\ell)$ were settled in Theorem
\ref{thm-liftedRMWD2}. 
\item It was proved in Theorem \ref{thm-last} and Corollary \ref{cor-july24design3} that that many infinite families of $3$-designs are supported by 
the lifted codes $\RM_2(r,m)(2|2^\ell)$.  New infinite families of $3$-designs were obtained by studying 
these lifted codes. 
\end{itemize}

The determination of the weight enumerator of a linear code is quite difficult in general. The settlement of 
the weight enumerator of a lifted code $\C(q|q^\ell)$ could be very difficult, even if the weight enumerator of the given code $\C$ is known.  It is much harder to determine the parameters of the designs studied in this paper. The reader is cordially invited to attack the six open problems and Conjecture 
\ref{conj-august31} presented in this paper. 

It is known that the permutation automorphism group $\PAut(\cH_{(2,m)})$ is the general linear group 
$\GL_m(\gf(2))$ \cite{Huffman98}.  It would be infeasible to characterise the permutation automorphism group $\PAut(\cH_{(q,m)})$ for $q>2$, as it varies from a Hamming code to another Hamming code of the same 
parameters, although all Hamming codes $\cH_{(q,m)}$ are monomially-equivalent.  In a special case, the 
permutation automorphism group $\PAut(\cH_{(q,m)})$ for $q>2$ was settled \cite{Gorkunov}.

\section*{Acknowledgements}

C. Ding's research was supported by the Hong Kong Research Grants Council,
Proj. No. 16301123 and in part by 
the UAEU-AUA joint research grant G00004614.  
Z. Sun’s research was supported by The National Natural Science Foundation of China under Grant Number 62002093. 
All the computations of this paper were done with the Magma 
software package. 


\end{document}